\documentclass[10pt,letterpaper]{article}
\usepackage[latin9]{inputenc}
\usepackage{units}
\usepackage{amsmath}
\usepackage{amsthm}
\usepackage{amssymb}
\usepackage{graphicx}

\makeatletter

\theoremstyle{plain}
\newtheorem{thm}{\protect\theoremname}
\theoremstyle{plain}
\newtheorem{prop}[thm]{\protect\propositionname}
\ifx\proof\undefined
\newenvironment{proof}[1][\protect\proofname]{\par
\normalfont\topsep6\p@\@plus6\p@\relax
\trivlist
\itemindent\parindent
\item[\hskip\labelsep\scshape #1]\ignorespaces
}{%
\endtrivlist\@endpefalse
}
\providecommand{\proofname}{Proof}
\fi
\theoremstyle{plain}
\newtheorem{cor}[thm]{\protect\corollaryname}

\usepackage[OE]{express}
\usepackage{cite}

\newcommand{\argmax}{\mathop{\mathrm{argmax}}}
\newcommand{\argmin}{\mathop{\mathrm{argmin}}}
\newcommand{\erfc}{\mathop{\mathrm{erfc}}\nolimits}

\usepackage[printonlyused]{acronym}
\acrodef{MAP}{maximum a posteriori probability}
\acrodef{ML}{maximum likelihood}
\acrodef{QAM}{quadrature amplitude modulation}
\acrodef{QPSK}{quadrature phase-shift keying}
\acrodef{ISI}{intersymbol interference}
\acrodef{GLME}{Gelfand-Levitan-Marchenko equation}
\acrodef{NFDM}{Nonlinear frequency-division multiplexing}
\acrodef{NFT}{nonlinear Fourier transform}
\acrodef{FNFT}{forward NFT}
\acrodef{BNFT}{backward NFT}
\acrodef{DF-BNFT}{decision-feedback BNFT}
\acrodef{OFDM}{orthogonal frequency-division multiplexing}
\acrodef{TX}{transmitter}
\acrodef{RX}{receiver}
\acrodef{FT}{Fourier transform}
\acrodef{DAC}{digital-to-analog converter}
\acrodef{ADC}{analog-to-digital converter}
\acrodef{LP}{Layer-Peeling}
\acrodef{SNR}{signal-to-noise ratio}
\acrodef{NIS}{Nonlinear inverse synthesis}
\acrodef{DBP}{digital backpropagation}
\acrodef{QAM}{quadrature amplitude modulation}
\acrodef{SNR}{signal to noise ratio}
\acrodef{SER}{symbol error rate}
\acrodef{SE}{spectral efficiency}
\acrodef{EDC}{electronic dispersion compensation}
\acrodef{NLSE}{nonlinear Schr\"odinger equation}
\acrodef{AWGN}{additive white gaussian noise}
\acrodef{GVD}{group velocity dispersion}

\makeatother

\providecommand{\corollaryname}{Corollary}
\providecommand{\propositionname}{Proposition}
\providecommand{\theoremname}{Theorem}

\begin{document}

\title{Decision-feedback detection strategy for nonlinear frequency-division
multiplexing}

\author{Stella Civelli,\authormark{1,2,{*}} Enrico Forestieri,\authormark{1,2} and Marco Secondini\authormark{1,2}}

\address{\authormark{1}TeCIP Institute, Scuola Superiore Sant'Anna, Pisa,
Italy\\
\authormark{2}Photonic Networks and Technologies Nat'l Lab, CNIT,
Pisa, Italy}

\email{\authormark{*}stella.civelli@santannapisa.it }
\begin{abstract}
By exploiting a causality property of the nonlinear Fourier transform,
a novel decision-feedback detection strategy for nonlinear frequency-division
multiplexing (NFDM) systems is introduced. The performance of the
proposed strategy is investigated both by simulations and by theoretical
bounds and approximations, showing that it achieves a considerable
performance improvement compared to previously adopted techniques
in terms of Q-factor. The obtained improvement demonstrates that,
by tailoring the detection strategy to the peculiar properties of
the nonlinear Fourier transform, it is possible to boost the performance
of NFDM systems and overcome current limitations imposed by the use
of more conventional detection techniques suitable for the linear
regime.
\end{abstract}


\bibliographystyle{osajnl}

\section{Introduction}

Nowadays, most of the global data traffic is carried by optical fibers.
However, optical fiber Kerr nonlinearity severely limits the capacity
of current optical communication systems, whose design is based on
concepts developed for linear systems\cite{turitsyn2017optica}. For
this reason, recently, there has been growing interest in \ac{NFT}-based
transmission schemes \cite{le2014nonlinear,le2016,turitsyn2017optica,tavakkolnia2016,Yousefi2014_NFT,Yousefi2016,bulow2015experimental}.
The \ac{NFT} \cite{Yousefi2014_NFT,ablowitz1981}, a sort of nonlinear
analogue of the standard \ac{FT}, defines a nonlinear spectrum that
evolves trivially and linearly along the optical channel on its nonlinear
frequency domain, thus turning nonlinearity into a mean for information
transfer rather than an impairment. \ac{NFDM} is a novel transmission
paradigm that aims at mastering nonlinearity using the \ac{NFT} to
avoid any deterministic intra and interchannel interference by encoding
information directly on the nonlinear spectrum \cite{Yousefi2014_NFT,Yousefi2016}.
However, research about \ac{NFDM} schemes is still ongoing and it
is not yet clear whether \ac{NFDM} can outperform conventional systems
\cite{civelli2017noise,menyuk2017NFT}. We refer to \cite{turitsyn2017optica}
for a complete review about different \ac{NFT}-based transmission
schemes.

\ac{NIS} \cite{le2014nonlinear} is a popular \ac{NFDM} scheme based
on the \ac{NFT} with vanishing boundary conditions. \ac{NIS} maps
the information on the continuous part of the nonlinear spectrum through
a \ac{BNFT}, and then recovers it with the reverse operation, the
\ac{FNFT}. An intrinsic limitation of this scheme with respect to
conventional systems was highlighted in a previous work \cite{civelli2017noise}.
Indeed, the insertion of guard symbols between different bursts, as
required by \ac{NIS}, causes a reduction of the overall spectral
efficiency that, differently from conventional systems, cannot be
mitigated by increasing the burst length, as system performance decays
for longer bursts. Moreover, in \cite{civelli2017noise} it is conjectured
that the \ac{NIS} performance decay is due to the considered symbol-by-symbol
detection strategy, based on the Euclidean distance, which is optimal
only in the linear regime and does not account for the statistics
of noise in the nonlinear spectrum.

In this work, we exploit a powerful property of the \ac{NFT} to devise
a novel detection strategy for \ac{NFDM} systems. The proposed strategy
employs decision feedback and the \ac{BNFT} to avoid the detrimental
effect that an increase of the burst length has on the noise statistics
in the nonlinear frequency domain. After introducing and discussing
the \ac{DF-BNFT} strategy, system performance obtained through simulations
are presented, and theoretical performance estimations are given.

The paper is organized as follows. Section~2 states and proves a
causality property of the \ac{NFT} that plays a key role in the derivation
of the \ac{DF-BNFT} detection strategy. Section~3 describes the
system and derives two important corollaries of the causality property.
Section~4 introduces and explains the \ac{DF-BNFT} detection strategy.
Section~5 derives an approximation, a lower bound, and an upper bound
to the probability of error of \ac{DF-BNFT}. Section~6 deals with
\ac{DF-BNFT} performance; also, it compares \ac{DF-BNFT} with standard
\ac{FNFT} detection and conventional systems in terms of performance,
considering different scenarios. Section~7 validates the approximation
and bounds derived in Section~5 by comparison with simulation results.
Finally, Section~8 concludes the paper.

\section{NFT causality property}

The \ac{BNFT}\textemdash by which a time domain signal $r(t)$ is
recovered from its nonlinear spectrum $\rho(\lambda)$\textemdash can
be computed via the \ac{GLME}\cite{ablowitz1981} 

\begin{equation}
K(x,y)-\sigma F^{*}(x+y)+\sigma\int_{x}^{\infty}\!\int_{x}^{\infty}\!K(x,r)F(r+s)F^{*}(s+y)\,\mathrm{d}r\mathrm{d}s=0\,,\label{eq:GLME}
\end{equation}
where $\sigma=1$ in the focusing regime and $\sigma=-1$ in the defocusing
regime, $^{*}$ denotes the complex conjugate, and $F(y)$ is an integral
function of the nonlinear spectrum \cite{ablowitz1981}. Specifically,
if there is no discrete spectrum, 
\begin{equation}
F(y)=\frac{1}{2\pi}\int_{-\infty}^{\infty}\rho(\lambda)e^{j\lambda y}\,\mathrm{d}\lambda.\label{eq:Fy}
\end{equation}
 The time domain signal is recovered from the GLME solution $K(x,y)$
as $r(t)=-2K(t,t)$. 
\begin{prop}
\emph{\label{prop:(Causality-Property)}(}\ac{NFT}\emph{ Causality
Property)} Given a generic time instant $\tau$, for $t\geq\tau$
the time domain signal $r(t)$ depends only on the values of $F(y)$
for $y\geq2\tau$. 
\end{prop}
\begin{proof}
Since $r(t)=-2K(t,t)$, one should obtain the solution $K(t,t)$ of
Eq.~(\ref{eq:GLME}) for all $t\geq\tau$. Given $t_{1},t_{2}\geq\tau$,
the solution $K(t_{1},t_{2})=\sigma F^{*}(t_{1}+t_{2})-\sigma\int_{t_{1}}^{\infty}\!\int_{t_{1}}^{\infty}\!K(t_{1},r)F(r+s)F^{*}(s+t_{2})\,\mathrm{d}r\mathrm{d}s$
depends on $F(y)$ for all $y\geq2\tau$ (since $r+s\geq2t_{1}\geq2\tau$
and $s+t_{2}\geq t_{1}+t_{2}\geq2\tau$) and on all the solutions
$K(t_{1},r)$ for $r\geq t_{1}$. Nevertheless, since $t_{1}\geq\tau$,
$K(t_{1},r)$ is one of the solutions $K(t_{1},t_{2})$ for $t_{1},t_{2}\geq\tau$.
Therefore, for $t_{1},t_{2}\geq\tau$, $K(t_{1},t_{2})$ depends only
on $F(y)$ for $y\geq2\tau$. Consequently, considering $t_{1}=t_{2}=t$,
one obtains that, for $t\geq\tau$, $r(t)=-2K(t,t)$ depends only
on $F(y)$ for $y\geq2\tau$.
\end{proof}

\section{System description\label{sec:System-description}}

\begin{figure}
\begin{centering}
\includegraphics[width=0.9\columnwidth]{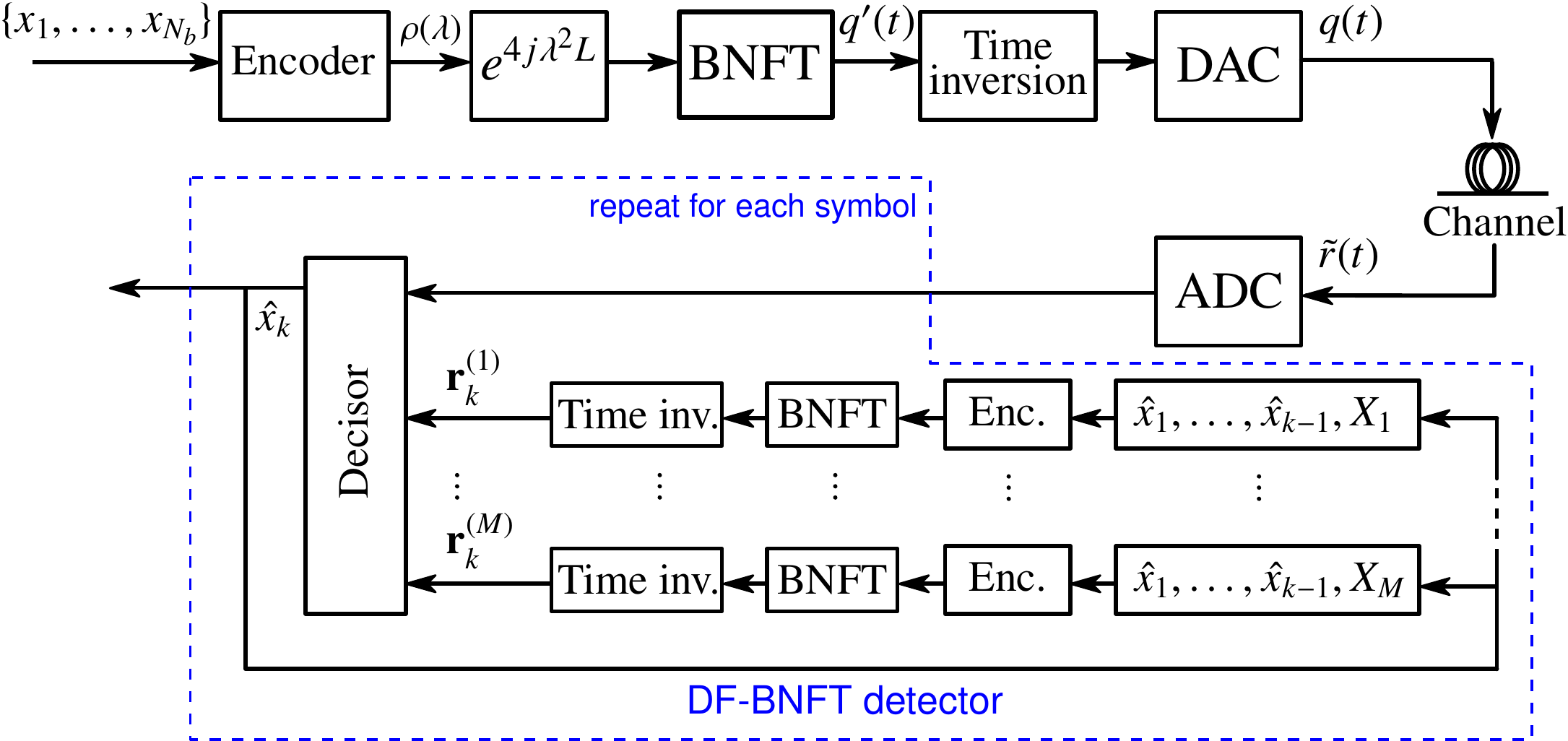}
\par\end{centering}
\caption{\label{fig:1}\protect\ac{NFDM} with \protect\ac{DF-BNFT} detection. }
\end{figure}
The transmission scheme considered in this work is sketched in Fig.~\ref{fig:1}.
As in the \ac{NIS} scheme \cite{le2014nonlinear}, the \ac{TX} encodes
a burst of $N_{b}$ symbols~$\{x_{1},\ldots,x_{N_{b}}\}$ drawn from
an $M$-ary \ac{QAM} alphabet $\{X_{1},\ldots,X_{M}\}$ onto a QAM
signal
\begin{equation}
s(t)=\sum_{k=1}^{N_{b}}x_{k}g[t-(k-1)T_{s}]\label{eq:QAM_signal}
\end{equation}
with pulse shape $g(t)$ and symbol time $T_{s}$. In this work,
for reasons that will be clarified later, we restrict $g(t)$ to have
a finite duration $T\le T_{s}$. The ordinary Fourier transform $S(f)$
of (\ref{eq:QAM_signal}) is then mapped on the continuous part of
the nonlinear spectrum $\rho(\lambda)$ according to
\begin{equation}
\rho(\lambda)=-S(-\lambda/\pi).\label{eq:NFT-mapping}
\end{equation}
Furthermore, before computing the \ac{BNFT}, deterministic propagation
effects (dispersion and nonlinearity) are precompensated by multiplying
the nonlinear spectrum by $\exp(j4\lambda^{2}L)$, where $L$ is the
link length. Finally, the input optical signal is taken to be $q(t)=q'(-t)$,
where $q'(t)$ is the \ac{BNFT} of the precompensated nonlinear spectrum.

There are two differences between the \ac{TX} described here and
the one in \cite{le2014nonlinear}, which, however, do not change
the overall \ac{NIS} working principle. Firstly, propagation effects
are removed at the \ac{TX} (precompensation) rather than at the \ac{RX}.
While both solutions are feasible (and splitting the compensation
between TX and RX might even be advantageous \cite{civelli2017noise}),
precompensation is considered here because the proposed \ac{DF-BNFT}
detection strategy relies on it. Indeed, Corollary \ref{cor:noISI},
on which \ac{DF-BNFT} is based (as explained in the following), applies
only to optical signals unaffected by propagation effects. Precompensation
ensures that the received signal meet this requirement.  Secondly,
in this work each burst in the optical signal is inverted in time
just to help explaining the working principle, not because it is necessary.
Finally, for the sake of simplicity, we avoid explicitly mentioning
the normalization and denormalization procedures required to relate
the normalized signals $q(t)$ and $\tilde{r}(t)$ (see Fig.\,\ref{fig:1})
to the physical signals propagating in a real fiber link, therefore
assuming that the channel is characterized by the normalized version
of the \ac{NLSE} given in \cite{Yousefi2014_NFT}.

Let us denote by $r(t)$ the optical signal obtained by propagating
$q(t)$ in a \emph{noiseless} channel. In this case, the only channel
effect is the multiplication of the nonlinear spectrum by $\exp(-j4\lambda^{2}L)$,
so that $r'(t)=r(-t)$ would be the BNFT of $\rho(\lambda)$, the
nonlinear spectrum before precompensation.
\begin{cor}
\label{cor:QAM-signal}Considering a \ac{NIS} modulation, the above
optical signal $r(t)$ for $t\leq\tau$ depends only on the values
of the QAM signal $s(t)$ for $t\leq\tau$. 
\end{cor}
\begin{proof}
Taking into account (\ref{eq:Fy}), (\ref{eq:QAM_signal}) and (\ref{eq:NFT-mapping}),
we have that $F(y)=-s(-y/2)/2$ and thus, $F(y)$ for $y\geq2\tau$
depends only on the values of $s(t)$ for $t\leq-\tau$. Therefore,
recalling that $r'(t)$ is the \ac{BNFT} of $\rho(\lambda)$, the
\ac{NFT} causality property (Proposition~\ref{prop:(Causality-Property)})
implies that $r'(t)$ for $t\geq\tau$ depends only on $s(t)$ for
$t\leq-\tau$. Changing the sign of $\tau$, one obtains that $r'(t)$
for $t\geq-\tau$ depends only on $s(t)$ for $t\leq\tau$. Finally,
the thesis follows considering that $r(t)=r'(-t)$.
\end{proof}
\begin{cor}
\label{cor:noISI}Let $t_{k}=(k-1/2)T_{s}$. If the pulse shape $g(t)$
in (\ref{eq:QAM_signal}) has finite duration $T\le T_{s}$, for $t\le t_{k}$
the optical signal $r(t)$ depends only on the symbols $x_{1},..,x_{k}$,
and not on the following ones $x_{k+1},..,x_{N_{b}}$, as shown in
Fig.~\ref{fig:2}. In mathematical formulas 
\begin{equation}
r(t)\bigl|_{t\leq t_{k}}=\mathcal{G}\left\{ x_{1},..,x_{k}\right\} ,\label{eq:NFTproperty}
\end{equation}
where $\mathcal{G}$ is a generic function.
\end{cor}
\begin{proof}
With this hypothesis, $s(t)$ depends only on the symbols $x_{1},\ldots,x_{k}$
for $t\leq t_{k}$. Then, applying Corollary~\ref{cor:QAM-signal},
the thesis follows. 
\end{proof}
\begin{figure}
\begin{centering}
\includegraphics[width=0.75\columnwidth]{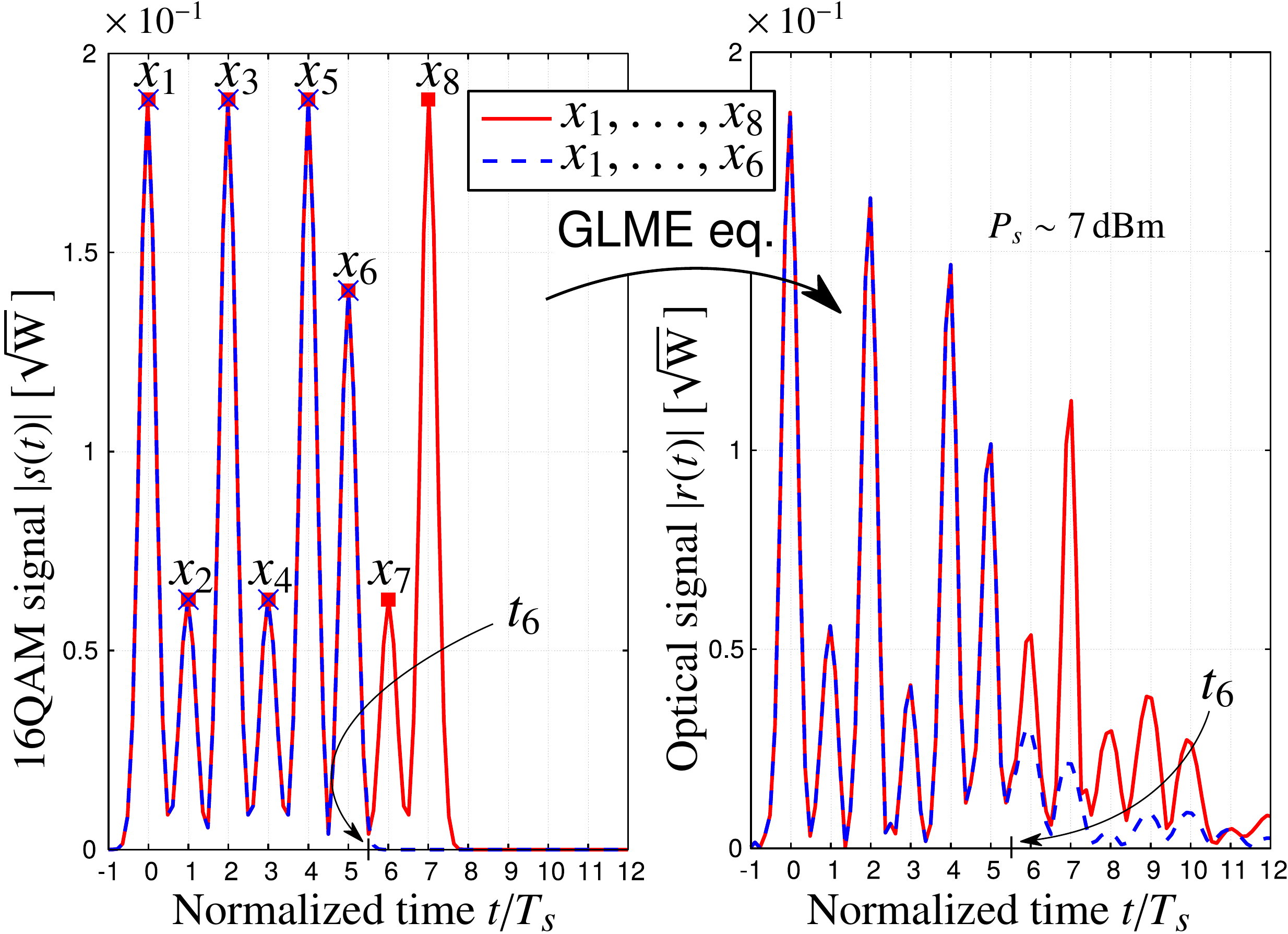}
\par\end{centering}
\caption{\label{fig:2}The \protect\ac{NFT} causality property for \protect\ac{NIS}
with no \protect\ac{ISI} on $s(t)$. A train of Gaussian pulses,
modulated by 16\protect\ac{QAM} symbols, and almost \protect\ac{ISI}-free,
is shown before (on the left) and after (on the right) the \protect\ac{BNFT}.
The red signal is generated by $8$ symbols, while for the blue one
only the first $6$ are taken into account. The two optical signals
are superimposed for $t\leq t_{6}$, as for Eq.~(\ref{eq:NFTproperty})
(baudrate $R_{s}=\unit[50]{GBd}$, optical power $P_{s}=\unit[7]{dBm}$).}
\end{figure}

\section{Decision-feedback BNFT detection\label{sec:Decision-feedback-BNFT-detection}}

In conventional \ac{NIS}, the \ac{RX} recovers a noisy version of
the transmitted nonlinear spectrum $\rho(\lambda)$ by computing the
\ac{FNFT} of the received optical signal, and then makes decisions
based on standard matched filtering and symbol-by-symbol detection.
The improved detection scheme proposed in this work originates from
the idea that, since a detrimental signal-noise interaction takes
place when computing the \ac{FNFT} of the received noisy signal,
decisions could be alternatively made by comparing the received signal
with the \ac{BNFT} of all possible transmitted (noiseless) waveforms,
thus avoiding signal-noise interaction effects. Selecting the waveform
(and the corresponding symbols) closest to the received optical signal
would correspond to a \ac{MAP} strategy, under the assumption that
the accumulated optical noise can be modeled as \ac{AWGN} (more on
this later). The drawback of a sequence\textemdash rather than symbol-by-symbol\textemdash detection
strategy, is an exponential growth of the detector complexity with
the burst length $N_{b}$. In order to avoid this growth, the \ac{NFT}
causality property and a decision-feedback scheme are finally employed,
obtaining the \ac{DF-BNFT} detection scheme depicted in Fig.~\ref{fig:1}.

Having in mind that \ac{NIS} performance decay is caused by the detection
strategy itself, rather than signal-noise interaction during propagation\textemdash as
also hinted by the fact that the performance obtained by simply adding
AWGN after noiseless propagation is superimposed with that of the
actual optical link \cite{civelli2017noise}\textemdash the aim of
this work is to devise a detection strategy at least optimal for the
\ac{AWGN} channel. Therefore, assuming the channel as \ac{AWGN},
we can write the received noisy optical signal as 
\begin{equation}
\tilde{r}(t)=r(t)+n(t),\label{eq:AWGN}
\end{equation}
 where $r(t)$ is, as before, the optical signal obtained by propagating
the input signal $q(t)$ in a \emph{noiseless} channel, and $n(t)$
is circularly-symmetric complex white Gaussian noise with power spectral
density $N_{0}$. We would like to stress that (\ref{eq:AWGN}) is
only an \emph{ansatz} and not an actual identity.

An \ac{ADC} recovers the samples of the received noisy optical signal
and collects them in the vector $\tilde{\boldsymbol{r}}$. The \ac{ADC}
is modeled as a rectangular filter with bandwidth $\nu/(2T_{s})$
that acquires $\nu$ samples per symbol time. Assuming that the filter
bandwidth is larger than the overall signal bandwidth, under the \ac{AWGN}
assumption (\ref{eq:AWGN}), $\tilde{\mathbf{r}}$ is a sufficient
statistic and we can write $\tilde{\boldsymbol{r}}=\boldsymbol{r}+\boldsymbol{n}$,
where $\boldsymbol{r}$ is a vector collecting the samples of $r(t)$,
and $\boldsymbol{n}$ is a vector of i.i.d. circularly-symmetric complex
Gaussian r.v.s $n_{k}$, with zero mean and variance $\sigma^{2}=E\{|n_{k}|^{2}\}=N_{0}\nu/T_{s}$.
Therefore, conditional on $\boldsymbol{r}$, the components of $\tilde{\boldsymbol{r}}$
are independent. For the sake of simplicity, let $\tilde{\boldsymbol{r}}_{k}$
(and $\boldsymbol{r}_{k}$) be the vector of length $\nu$ representing
the noisy optical signal (and, respectively, its noise-free equivalent)
in the time window $[t_{k-1},t_{k})$. Hence, $\tilde{\boldsymbol{r}}$
can also be written as a compound vector $\tilde{\boldsymbol{r}}=(\tilde{\boldsymbol{r}}_{1},\ldots,\tilde{\boldsymbol{r}}_{N_{b}})$
containing the samples of the received signal in $[-T_{s}/2,(N_{b}-1/2)T_{s})$,
i.e., the time window of duration $N_{b}T_{s}$ in which information
is encoded. Indeed, with respect to the \ac{QAM} signal $s(t)$,
the optical signal after the \ac{GLME} broadens in time, developing
a sort of right tail that extends outside the considered detection
window, i.e., for $t>(N_{b}-1/2)T_{s}$, as shown for instance in
Fig.~\ref{fig:2}. Therefore, a longer vector $\tilde{r}$ should
be considered to obtain a sufficient statistic. However, in the decision-feedback
strategy derived in the following, this tail could be used only to
detect the last symbol $x_{N_{b}}$ of the sequence, with a negligible
contribution to the overall performance. Hence, for the sake of simplicity,
we simply discard it. 

According to the \ac{MAP} strategy, and assuming equally likely input
symbols, optimal detection maximizes the probability density function
(pdf) $p(\tilde{\boldsymbol{r}}|\boldsymbol{x})$ of the vector $\tilde{\boldsymbol{r}}$
conditional upon the transmitted sequence $\boldsymbol{x}=(x_{1},..,x_{N_{b}})$.
Thus, an optimum \ac{RX} chooses the sequence $\hat{\boldsymbol{x}}$
according to

\begin{equation}
\hat{\boldsymbol{x}}=\argmax_{\boldsymbol{x}}p(\tilde{\boldsymbol{r}}|\boldsymbol{x}).\label{eq:MAP}
\end{equation}
Since, conditional upon $\boldsymbol{x}$, the components of $\tilde{\boldsymbol{r}}$
are independent, the pdf in Eq.~(\ref{eq:MAP}) can be factorized
as 
\begin{equation}
\hat{\boldsymbol{x}}=\argmax_{\boldsymbol{x}}\prod_{k=1}^{N_{b}}p(\tilde{\boldsymbol{r}}_{k}|\boldsymbol{x})=\argmax_{\boldsymbol{x}}\sum_{k=1}^{N_{b}}\ln p(\tilde{\boldsymbol{r}}_{k}|\boldsymbol{x}),\label{eq:NNb}
\end{equation}
where the second equality stems from the monotonic behavior of the
logarithm. 

The \ac{NFT} property (\ref{eq:NFTproperty}) implies that the signal
samples in the time window $[t_{k-1},t_{k})$, i.e., those collected
in $\tilde{\boldsymbol{r}}_{k}$, depend only on the symbols $(x_{1},\ldots,x_{k})$.
Therefore, an optimum RX performs decisions according to

\begin{equation}
\hat{\boldsymbol{x}}=\argmax_{\boldsymbol{x}}\sum_{k=1}^{N_{b}}\mathrm{ln}\,p(\tilde{\boldsymbol{r}}_{k}|(x_{1},\ldots,x_{k})).\label{eq:optimum_detection}
\end{equation}
Since the symbols $x_{1},\ldots,x_{k}$ uniquely determine $\boldsymbol{r}_{k}$,
under the AWGN assumption we have
\begin{equation}
p(\tilde{\boldsymbol{r}}_{k}|(x_{1},\ldots,x_{k})=\frac{1}{(\pi\sigma^{2})^{\nu}}\exp\left(-\Vert\tilde{\boldsymbol{r}}_{k}-\boldsymbol{r}_{k}\Vert^{2}/\sigma^{2}\right)\label{eq:pdfrk}
\end{equation}
so that $\mathrm{ln}\,p\bigl(\tilde{\boldsymbol{r}}_{k}\big|(x_{1},\ldots,x_{k})\bigl)=-\nu\,\mathrm{ln}\left(\pi\sigma^{2}\right)-\Vert\tilde{\boldsymbol{r}}_{k}-\boldsymbol{r}_{k}\Vert^{2}/\sigma^{2}$.
This implies that, in order to determine the optimal sequence $\hat{\boldsymbol{x}}$,
all possible $M^{N_{b}}$ input sequences should be considered, which
is not a viable solution. In order to avoid this exponential growth
of complexity, we resort to a sub-optimal decision-feedback strategy\textemdash namely,
\ac{DF-BNFT}\textemdash by which symbols are decided iteratively
for $k=1,\ldots,N_{b}$ as
\begin{equation}
\hat{x}_{k}=\argmax_{X_{i}\in\left\{ X_{1},..,X_{M}\right\} }\mathrm{ln}\,p\bigl(\tilde{\boldsymbol{r}}_{k}|(\hat{x}_{1},...,\hat{x}_{k-1},X_{i})\bigr)\label{eq:suboptimal}
\end{equation}
bringing down to $M\times N_{b}$ the number of sequences to be considered.
Equation (\ref{eq:suboptimal}) can be evaluated comparing the (samples
of the) received signal with (the samples of) $M$ trial waveforms
$r_{k}^{(i)}(t)$ uniquely corresponding to the sequences $\hat{x}_{1},...,\hat{x}_{k},X_{i}$
for each $X_{i}$ in the symbol constellation $\{X_{1},..,X_{M}\}$.
The waveform $r_{k}^{(i)}(t)$ is obtained from the symbol sequence
$\hat{x}_{1},...,\hat{x}_{k},X_{i}$ by the same encoding technique
used at the \ac{TX}, except for precompensation. Let $\boldsymbol{r}_{k}^{(i)}$
denote the vector of length $\nu$ containing the samples of $r_{k}^{(i)}(t)$
in $[t_{k-1},t_{k})$, referred to as detection window. In other words,
$\boldsymbol{r}_{k}^{(i)}$, $i=1,\ldots,M$, are all possible vectors
$\boldsymbol{r}_{k}$ given that the sequence $\hat{x}_{1},...,\hat{x}_{k},X_{i}$
has been sent. The \ac{RX} implements the \ac{DF-BNFT} strategy
through $N_{b}$ steps as follows.

For $k=1,\ldots,N_{b}$: 
\begin{itemize}
\item Digitally obtain the vectors $\boldsymbol{r}_{k}^{(i)}$, for each
$X_{i}$ in the symbol constellation $\{X_{1},..,X_{M}\}$. 
\item Choose $\hat{x}_{k}$ by the rule
\begin{equation}
\hat{x}_{k}=\argmin_{X_{i}\in\{X_{1},..,X_{M}\}}\bigl\Vert\tilde{\boldsymbol{r}}_{k}-\boldsymbol{r}_{k}^{(i)}\bigr\Vert^{2}\label{eq:kdec}
\end{equation}
where this last equation follows from (\ref{eq:pdfrk}) and (\ref{eq:suboptimal}). 
\end{itemize}
The \ac{DF-BNFT} strategy (\ref{eq:suboptimal}) avoids any \ac{ISI}
thanks to decision feedback, which accounts for the dependence of
$\tilde{\mathbf{r}}_{k}$ on previous symbols, and to Corollary~\ref{cor:noISI},
which ensures that $\tilde{\mathbf{r}}_{k}$ does not depend on next
symbols. However, it is still suboptimal compared to (\ref{eq:optimum_detection})
for two reasons. The first reason is that it does not exploit the
information about $x_{k}$ that is contained in the received signal
after $t_{k}$. In fact, as shown in Fig.~\ref{fig:2}, while in
the original \ac{QAM} signal $s(t)$ (on the left) the information
about $x_{k}$ is fully contained in the time interval $[t_{k-1},t_{k})$,
in the corresponding optical signal $r(t)$ (on the right) part of
this information goes to times $t>t_{k}$, the effect becoming more
relevant with power and as $k$ increases. An apparent consequence
of this effect is that the optical signal $r(t)$ has an average amplitude
that decreases with time and a sort of ``tail'' that extends beyond
the duration of the original \ac{QAM} signal. The second reason is
that it is affected by error propagation, as previous decisions $\hat{x}_{1},...,\hat{x}_{k-1}$
might be incorrect, this effect becoming more relevant as $k$ increases,
too. Finally, even the strategy (\ref{eq:optimum_detection}), that
is optimal for the \ac{AWGN} channel (\ref{eq:AWGN}), might be suboptimal
on a real fiber link. The effects of error propagation, information
loss, and non-\ac{AWGN} channel statistics are discussed in more
detail in Section~\ref{sec:Numerical-results}.

While the \ac{DF-BNFT} strategy is conceived to improve system performance
(as it will be verified in Section~\ref{sec:Numerical-results}),
it also has some drawbacks compared to the \ac{FNFT} strategy. Firstly,
 to fulfill the hypothesis of Corollary~\ref{cor:noISI} and avoid
ISI, the pulse shape must be fully confined within a symbol time\textemdash a
more stringent requirement compared to the conventional Nyquist criterion.
This imposes the use of pulses with a wider bandwidth, which has a
negative impact on the achievable spectral efficiency. Secondly, full
channel precompensation is required to use Corollary~\ref{cor:noISI}
at the \ac{RX}, and, therefore, channel compensation cannot be split
between \ac{RX} and \ac{TX} to reduce guard intervals \cite{civelli2017noise,tavakkolnia2016}.
Thirdly, as far as it concerns computational complexity, the \ac{DF-BNFT}
scheme requires the evaluation of $M$ \ac{BNFT}s over $\nu N_{b}$
points. On the other hand, standard NIS detection computes only one
\ac{FNFT} over a comparable number of points. While an exact comparison
between the two strategies depends on the relative complexity and
accuracy of the considered \ac{BNFT} and \ac{FNFT} algorithms, it
is reasonable to assume that the complexity of the \ac{DF-BNFT} detector
is considerably higher than that of standard \ac{FNFT} detection.
However, the aim of this work is to show that standard \ac{FNFT}
detection is far from being optimal, and that a significant increase
of performance can be achieved by an improved detection strategy.

\section{Error probability estimation and bounds}

For a given sequence, the probability of error $P_{e}$ of the \ac{DF-BNFT}
strategy can be evaluated by averaging (over all constellation symbols
and symbols within a burst) the probability $P_{k}^{(m)}$ that an
error occurs when the symbol $X_{m}$ is sent at position $k$, provided
that the symbols $\hat{x}_{1},..\hat{x}_{k-1}$ are correctly detected.
Specifically, 

\begin{equation}
P_{e}=\frac{1}{MN_{b}}\sum_{k=1}^{N_{b}}\sum_{m=1}^{M}P_{k}^{(m)}.\label{eq:Pe}
\end{equation}

As computing $P_{k}^{(m)}$ may be difficult, we will bound it through
standard techniques \cite{Proakis5th_digitalcomm}. By defining the
events $E_{m,i}=\{\text{\ensuremath{X_{i}} is preferred to \ensuremath{X_{m}} when deciding on \ensuremath{x_{k}}}\}$
and
\begin{equation}
E_{m}=\{\text{\ensuremath{X_{m}}\ is not preferred when deciding on \ensuremath{x_{k}}}\}=\bigcup\nolimits _{\begin{subarray}{c}
i=1\\
i\ne m
\end{subarray}}^{M}E_{m,i}
\end{equation}
we can upper bound $P_{k}^{(m)}$ by the union bound
\begin{align}
P_{k}^{(m)}=P\left(E_{m}\bigm|(\hat{x}_{1},...,\hat{x}_{k-1},X_{m})\right) & =P\left(\bigcup\nolimits _{\begin{subarray}{c}
i=1\\
i\ne m
\end{subarray}}^{M}E_{m,i}\Bigm|(\hat{x}_{1},...,\hat{x}_{k-1},X_{m})\right)\nonumber \\
 & \le\sum_{\begin{subarray}{c}
i=1\\
i\neq m
\end{subarray}}^{M}P\left(E_{m,i}\Bigm|(\hat{x}_{1},...,\hat{x}_{k-1},X_{m})\right)\label{eq:ubound}
\end{align}
Due to our AWGN assumption, the pairwise error probabilities in (\ref{eq:ubound})
are given by
\begin{equation}
P\left(E_{m,i}\Bigm|(\hat{x}_{1},...,\hat{x}_{k-1},X_{m})\right)=\mathcal{Q}\left(\frac{d_{k}^{(m,i)}}{2\sigma}\right)\label{eq:pairwise}
\end{equation}
where 
\begin{equation}
d_{k}^{(m,i)}=\bigl\Vert\boldsymbol{r}_{k}^{(m)}-\boldsymbol{r}_{k}^{(i)}\bigr\Vert\label{eq:dk}
\end{equation}
is the Euclidean distance between $\boldsymbol{r}_{k}^{(m)}$ and
$\boldsymbol{r}_{k}^{(i)}$, and $\mathcal{Q}(x)\triangleq\int_{x}^{\infty}e^{-t^{2}/2}\,\mathrm{d}t/\sqrt{2\pi}=\erfc(x/\sqrt{2})/2$
is the Q-function. We can also obtain a useful approximation on $P_{k}^{(m)}$
as follows. Denoting by $C_{m,i}$ the event complementary to $E_{m,i}$,
we have
\begin{equation}
P_{k}^{(m)}=P\left(\bigcup\nolimits _{\begin{subarray}{c}
i=1\\
i\ne m
\end{subarray}}^{M}E_{m,i}\Bigm|(\hat{x}_{1},...,\hat{x}_{k-1},X_{m})\right)=1-P\left(\bigcap\nolimits _{\begin{subarray}{c}
i=1\\
i\ne m
\end{subarray}}^{M}C_{m,i}\Bigm|(\hat{x}_{1},...,\hat{x}_{k-1},X_{m})\right)
\end{equation}
and, taking into account that $P\bigl(C_{m,i}\bigm|(\hat{x}_{1},...,\hat{x}_{k-1},X_{m})\bigr)=1-P\bigl(E_{m,i}\bigm|(\hat{x}_{1},...,\hat{x}_{k-1},X_{m})\bigr)$,
\begin{equation}
P_{k}^{(m)}\simeq1-\prod_{\begin{subarray}{c}
i=1\\
i\neq m
\end{subarray}}^{M}\Bigl(1-P\left(E_{m,i}\bigm|(\hat{x}_{1},...,\hat{x}_{k-1},X_{m})\right)\Bigr)\label{eq:app}
\end{equation}
where the approximation is due to the fact that, in general, the events
$C_{m,i}$ are not mutually independent. Let us now derive a lower
bound. Recalling that the probability of a union of events is lower
bounded by each one of the probabilities of the single events, we
have
\begin{equation}
P_{k}^{(m)}=P\left(\bigcup\nolimits _{\begin{subarray}{c}
i=1\\
i\ne m
\end{subarray}}^{M}E_{m,i}\Bigm|(\hat{x}_{1},...,\hat{x}_{k-1},X_{m})\right)\ge\max_{i\ne m}P\left(E_{m,i}\bigm|(\hat{x}_{1},...,\hat{x}_{k-1},X_{m})\right).\label{eq:lbound}
\end{equation}
In conclusion, from (\ref{eq:ubound}), (\ref{eq:app}), (\ref{eq:lbound}),
and taking into account (\ref{eq:pairwise}), we have the following
upper bound, approximation and lower bound, respectively, on $P_{k}^{(m)}$
\begin{align}
P_{k}^{(m)} & \le\sum_{\begin{subarray}{c}
i=1\\
i\neq m
\end{subarray}}^{M}\mathcal{Q}\left(\frac{d_{k}^{(m,i)}}{2\sigma}\right)\label{eq:uppbound}\\
P_{k}^{(m)} & \simeq1-\prod_{\begin{subarray}{c}
i=1\\
i\neq m
\end{subarray}}^{M}\left(1-\mathcal{Q}\left(\frac{d_{k}^{(m,i)}}{2\sigma}\right)\right)\label{eq:approx}\\
P_{k}^{(m)} & \ge\mathcal{Q}\left(\frac{d_{k}}{2\sigma}\right)\label{eq:lowbound}
\end{align}
where 
\begin{equation}
d_{k}=\min_{i\neq m}d_{k}^{(m,i)}
\end{equation}
and $d_{k}^{(m,i)}$ is as in (\ref{eq:dk}). Replacing (\ref{eq:uppbound})\textendash (\ref{eq:lowbound})
into (\ref{eq:Pe}) gives the corresponding bounds and approximation
on $P_{e}$. We remark that both the bounds and the approximation
were derived by assuming that the channel is AWGN and that previous
decisions are correct. Deviations from this ideal situation might
invalidate the bounds and slightly reduce the accuracy of the approximation,
as shown in Section~\ref{sec:Validation-of-the-bounds} by numerical
simulations.

The above estimates of the probability of error for a given sequence
should be averaged over all possible $M^{N_{b}}$ sequences to obtain
the average error probability. However, to speed up computation, one
can think of averaging over randomly generated sequences until the
result stabilizes. As we will show in Section~\ref{sec:Validation-of-the-bounds},
this practical approach still provides a reasonable accuracy and a
significant computational saving compared to direct error counting.

\section{System performance\label{sec:Numerical-results}}

We simulated the system described in the previous sections and sketched
in Fig.\,\ref{fig:1} by using a 16QAM signal $s(t)$ with symbol
rate $R_{s}=1/T_{s}=\unit[50]{GBd}$. To fulfill the no-\ac{ISI}
requirement on $s(t)$, the supporting pulse was chosen to be $g(t)=\exp\bigl(-12.5(t/T_{s})^{2}\bigr)$,
i.e., a Gaussian pulse with a full width at half maximum (FWHM)
of $(2/5)\sqrt{2\ln2}\,T_{s}\simeq T_{s}/2$, so that about $99.9\%$
of the energy of a pulse is contained in a symbol time. This pulse
shape is chosen to fulfill the requirement of Corollary \ref{cor:noISI}
on the duration of $g(t)$ with only a moderate bandwidth increase
compared to the root-raised-cosine pulse employed in \cite{civelli2017noise}.
Note that, by relaxing the energy constraint, the bandwidth could
be further reduced to approach the one in \cite{civelli2017noise}.
This, however, would also reduce the performance due to the ISI generated
by the pulse tail, with an overall effect on the achievable spectral
efficiency that should be carefully considered. The optimization of
the pulse shape for the maximization of the spectral efficiency is
outside the scope of this work and will be addressed in a future publication.
 A typical example of generated signal is shown in Fig.~\ref{fig:2}
on the left. 

The channel is a standard single-mode fiber of length $L=\unit[2000]{km}$
with \ac{GVD} parameter $\beta_{2}=\unit[-20.39]{ps^{2}/km}$, nonlinear
coefficient $\gamma=\unit[1.22]{W^{-1}km^{-1}}$, and attenuation
$\alpha=\unit[0.2]{dB/km}$. Ideal distributed amplification (spontaneous
emission factor $\eta_{\mathrm{sp}}=4$) is considered along the channel.
The bandwidth of both the \ac{DAC} and the \ac{ADC} is $\unit[100]{GHz}$.
In order to account for the \ac{NFT} boundary conditions and for
temporal broadening due to dispersion, a total of $N_{z}=2000$ guard
symbols is considered in our simulations. The \ac{FNFT} and \ac{BNFT}
operations are performed considering an oversampling factor of $8$
samples per symbols, respectively using the \ac{LP} method \cite{turitsyn2017optica}
and an enhanced version of the Nystrom method \cite{civelliNFT}.
The simulation results are deemed free from numerical inaccuracies,
since it was verified that the noise-free performance was sufficiently
higher than the noisy one (when applicable), and that higher accuracy
did not change the results, similarly to what done in \cite{Civelli_fotonica16}.

As often customary in optical communications, performance is measured
in terms of Q-factor, defined as $Q_{\mathrm{dB}}^{2}=20\log_{10}[\sqrt{2}\erfc^{-1}(2P_{b})]$,
where the probability of bit error $P_{b}$ is given by direct error
counting \cite{Grellier11}. The rate efficiency $\eta=N_{b}/(N_{b}+N_{z})$
is considered to account for the spectral efficiency loss due to the
insertion of guard symbols \cite{civelli2017noise}. The performance
is reported as a function of the mean power per symbol, defined as
$P_{s}=E_{s}/T_{s}$, where $E_{s}=E_{\mathrm{tot}}/N_{b}$ is the
mean energy per information symbol and $E_{\mathrm{tot}}$ the total
energy of the optical signal (so that the actual average optical power
is $\eta P_{s}$).

\begin{figure}
\begin{centering}
\includegraphics[width=0.6\columnwidth]{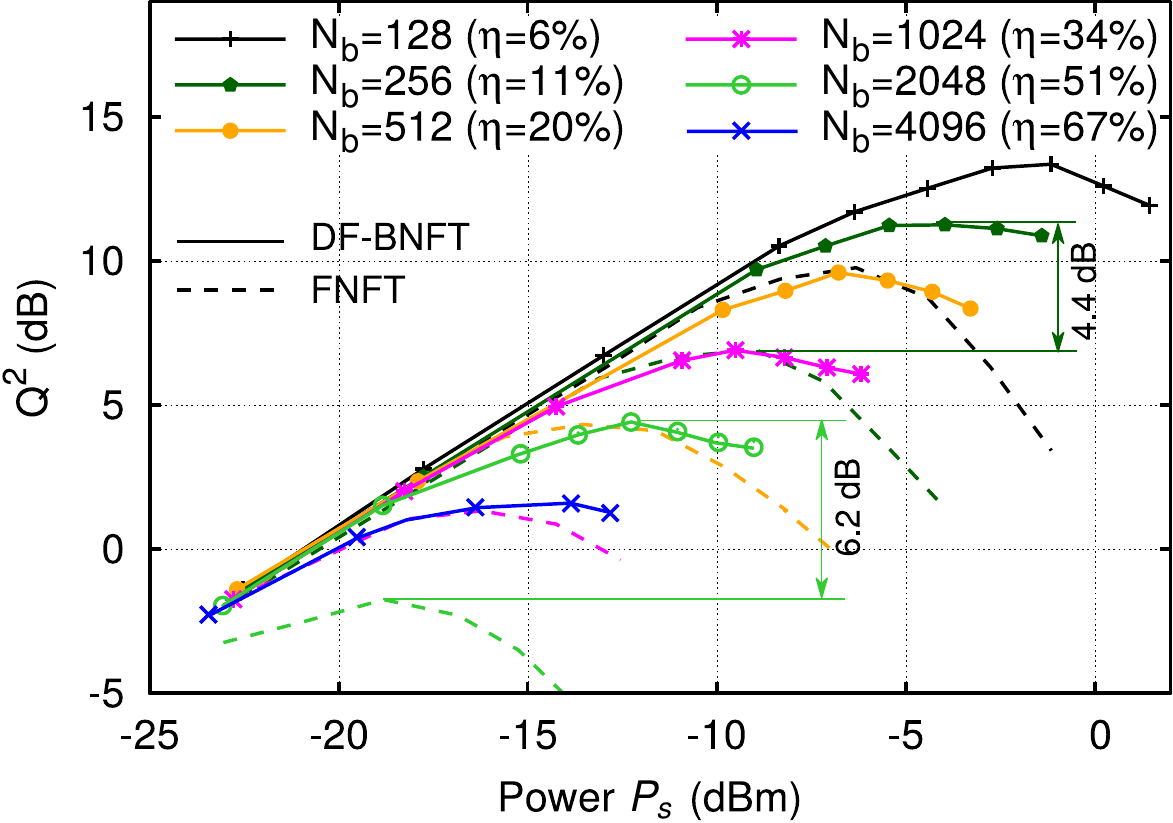}\vspace*{-1ex}
\par\end{centering}
\caption{\label{fig:3}Performance of the NFDM system for \protect\ac{DF-BNFT}
(solid lines) and standard FNFT (dashed lines) detection for different
burst length $N_{b}$ (and rate efficiency $\eta$). $16$\protect\ac{QAM}
symbols, $\beta_{2}=\unit[-20.39]{ps^{2}/km}$, $N_{z}=2000$, $L=\unit[2000]{km}$,
and $R_{s}=\unit[50]{GBd}$.}
\vspace*{-2ex}
\end{figure}

Figure\,\ref{fig:3} shows the \ac{NFDM} performance obtained with
\ac{DF-BNFT} detection (solid lines), and with conventional \ac{FNFT}
(dashed lines), for different burst lengths (same color for same length).
As can be seen, the performance obtained with \ac{DF-BNFT} is significantly
better than that obtained with \ac{FNFT} detection, with an improvement
of $4.4$\,dB for $N_{b}=256$ and $6.2$\,dB for $N_{b}=2048$.
However, performance still decays when increasing $N_{b}$. This behavior
may be due either to the suboptimality of the \ac{DF-BNFT} detection,
or to an intrinsic limitation of the NIS modulation format. For what
concerns suboptimality, there are three possible causes of performance
degradation, already discussed in Section~\ref{sec:Decision-feedback-BNFT-detection}:
the non-\ac{AWGN} statistics of the fiber channel, which is affected
by signal-noise interaction; the error propagation in the decision-feedback
mechanism of (\ref{eq:suboptimal}); and the information loss entailed
by (\ref{eq:kdec}), which neglects the information about $x_{k}$
that is contained in the received signal for $t>t_{k}$. All these
effects become more relevant as the burst length increases. Also
signal-noise interaction increases with the burst length, as the optical
noise interacts with a longer portion of non-zero signal. This effect,
however, saturates when the burst length becomes longer than the channel
memory. One may wonder whether the \ac{NIS} performance shown here
is in accordance with the theoretical estimation of the \ac{SNR}
given in \cite{Turitsyn_nature16}. Such a comparison was performed
in \cite{civelli2017noise} for almost the same system configuration
considered here. The only differences are: (1) the modulation format,
which however does not significantly change the results, and (2) the
chosen pulse shape. In this work, a pulse with a shorter time duration
and, hence, a wider bandwidth is considered. This difference is responsible
for a slight performance improvement.

\begin{figure}
\begin{centering}
\includegraphics[width=0.6\columnwidth]{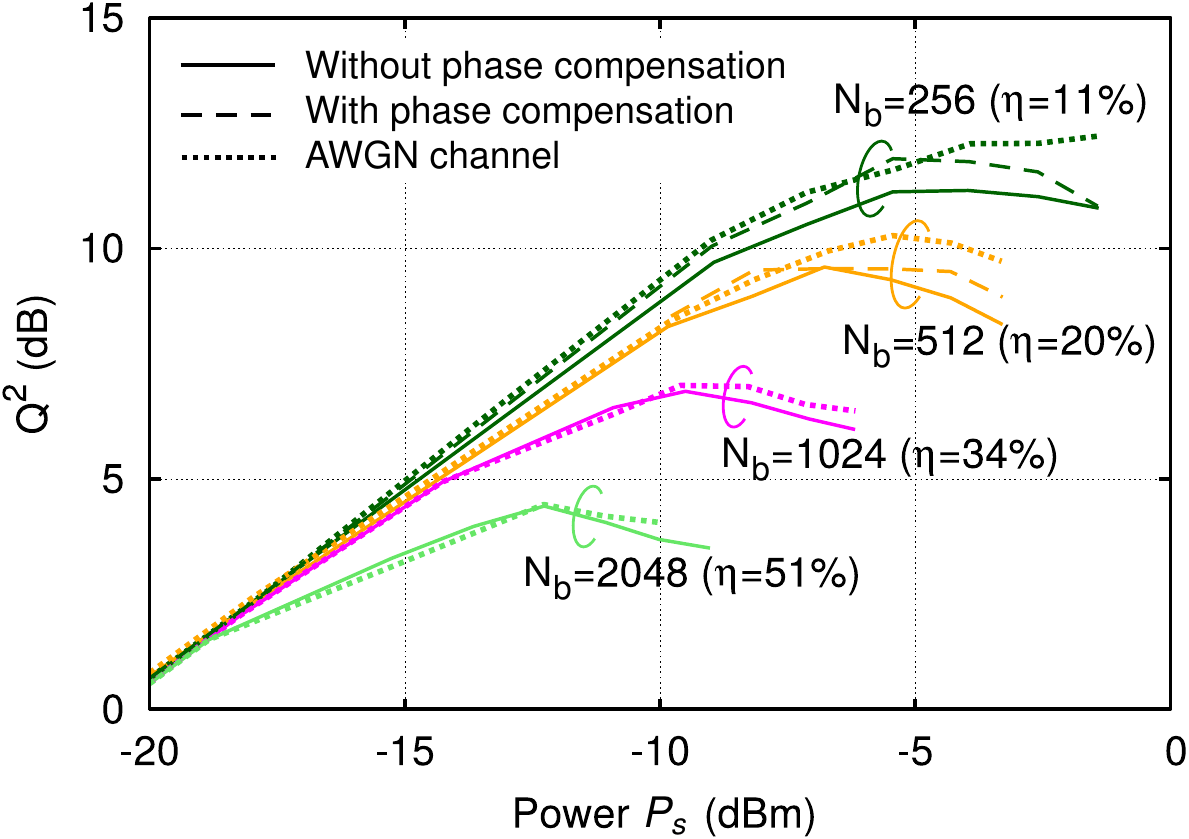}\vspace*{-1ex}
\par\end{centering}
\caption{\label{fig:4}Impact of fiber propagation: performance of DF-BNFT
on the fiber link without (solid lines) and with average nonlinear
phase compensation (dashed lines) and on the \protect\ac{AWGN} channel
(dotted line). Same scenario of Fig.~\ref{fig:3}.}
\end{figure}

The impact of signal-noise interaction during propagation can be estimated
by comparing the performance obtained on the fiber link with the performance
obtained on the corresponding \ac{AWGN} channel, shown in Fig.\,\ref{fig:4}
with dotted lines for $N_{b}=256,\,512,\,1024,\,2048$. For longer
bursts, i.e., $N_{b}\geq1024$, the corresponding dotted and solid
lines are almost indistinguishable. On the other hand, for shorter
bursts, i.e., $N_{b}=256,\,512$, the slight difference between the
dotted and solid lines denotes a small impact of signal-noise interaction
on system performance and a slight deviation of channel statistics
from the \ac{AWGN} assumption (\ref{eq:AWGN}). One of the effects
of signal-noise interaction during propagation is a constant phase
rotation of the optical signal. This deviation can be estimated and
removed from the optical signal considering for detection $\tilde{r}(t)e^{-j\alpha}$,
$\alpha$ being the phase shift, rather than $\tilde{r}(t)$ itself.
Fig.\,\ref{fig:4} shows that a small performance gain can be obtained
with this technique (performance shown with dashed lines) and that
the performance approaches those of the equivalent \ac{AWGN} channel.
 Obviously, when performance is superimposed to that of the \ac{AWGN}
channel, this technique does not affect the detection strategy.

\begin{figure}
\begin{centering}
\includegraphics[width=0.6\columnwidth]{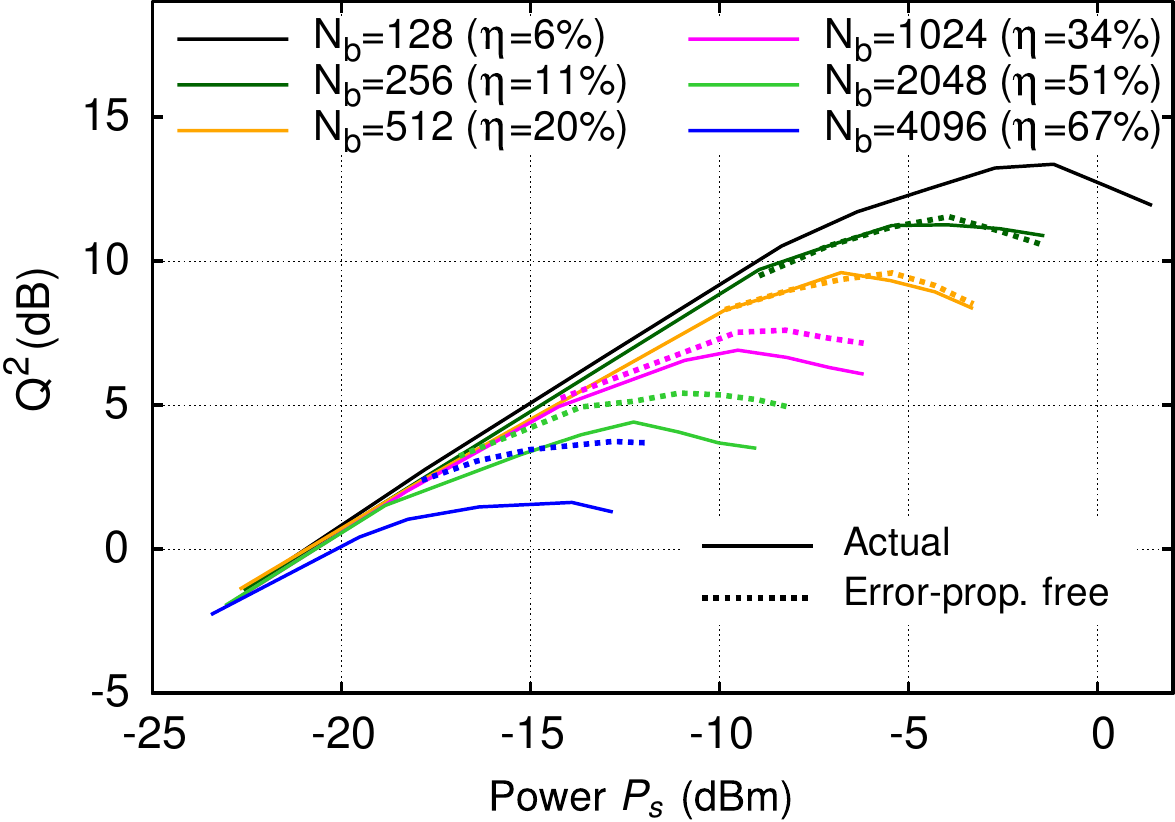}\vspace*{-1ex}
\par\end{centering}
\caption{\label{fig:5}Impact of error propagation due to decision feedback
in the proposed \protect\ac{DF-BNFT} detection strategy: actual performance
(solid lines), and error-propagation-free performance (dotted lines).
Same scenario of Fig.~\ref{fig:3}.}
\end{figure}

As regards error propagation, its impact can be estimated from Fig.\,\ref{fig:5},
in which the actual \ac{DF-BNFT} performance is compared to that
of an ideal detector that makes decisions according to the same strategy
(\ref{eq:suboptimal}), but using the correct symbols $x_{1},\ldots,x_{k-1}$
rather than the detected ones $\hat{x}_{1},\ldots,\hat{x}_{k-1}$.
Contrarily to what observed for signal-noise interaction, the impact
of error propagation is more relevant for longer bursts, while it
tends to be negligible for shorter ones. Indeed, for longer bursts,
farther symbols \emph{(i)} affect more significantly detection, and
\emph{(ii)} are more likely to be wrong, since the probability of
error is higher.

As regards the third possible cause of performance degradation, further
investigations are required to estimate the impact of information
loss due to the suboptimality of (\ref{eq:suboptimal}) and to devise
a better strategy to avoid it. Eventually, the implementation of an
optimal strategy based on (\ref{eq:optimum_detection}), but with
a feasible complexity, would allow to estimate the ultimate performance
of \ac{NIS} modulation and to understand if the observed performance
decay is due to suboptimal detection or to an intrinsic limitation
of this modulation scheme. This is currently under investigation.

\begin{figure}
\begin{centering}
~~~~\includegraphics[width=0.62\columnwidth]{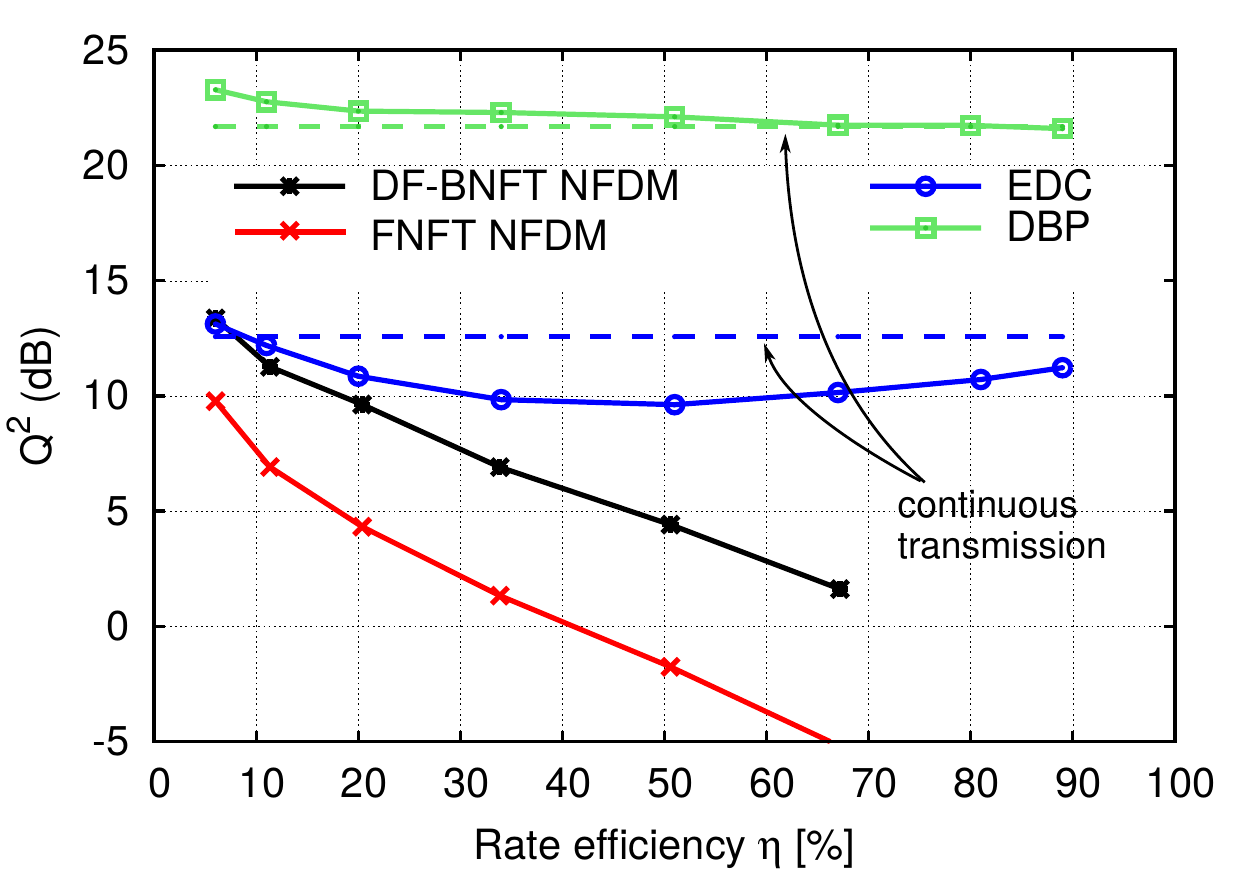}\vspace*{-1ex}
\par\end{centering}
\caption{\label{fig:6}Best achievable performance vs rate efficiency for \protect\ac{NFDM}
with different detection strategies and for conventional systems with
EDC or DBP. Same scenario of Fig.~\ref{fig:3}.}
\vspace*{-2ex}
\end{figure}

The maximum performance achieved by \ac{FNFT} and \ac{DF-BNFT} detection
(at their respective optimum power) in Fig.\,\ref{fig:3} are reported
in Fig.\,\ref{fig:6} as a function of the rate efficiency and compared
with the maximum performance achieved by conventional systems (also
operating in burst mode, for a fair comparison) employing ideal \ac{EDC}
and \ac{DBP} (practically implemented by the split-step Fourier method
with $100$ steps per span of fiber, enough to practically achieve
a perfect compensation of deterministic nonlinearity). \ac{DBP} performance
is estimated from the error vector magnitude \cite{shafik2006extended},
rather than calculated by direct error counting, being the corresponding
error probability too low to be measured. The improvement of \ac{DF-BNFT}
with respect to \ac{FNFT} is quite relevant and slightly increases
with the rate efficiency $\eta$. However, the \ac{DF-BNFT} performance
is still not on par with that of conventional systems and keeps decreasing
at higher rates, when the performance of conventional systems saturate
to the one achieved with continuous (non burst mode) transmission.
Finally, the performance achievable by \ac{DF-BNFT} detection was
investigated also in different scenarios. Fig.\,\ref{fig:7}a refers
to the same system setup used in Figs.\,\ref{fig:3}\textendash \ref{fig:6}
but with a lower dispersion parameter $\beta_{2}=\unit[-1.27]{ps^{2}/km}$
and, therefore, a lower number of guard symbols $N_{z}=125$. The
overall results do not change significantly, as already observed in
\cite{civelli2017noise}, and \ac{DF-BNFT} achieves a performance
improvement of almost $6$\,dB with respect to \ac{FNFT} detection.
The behavior also does not change when using \ac{QPSK} symbols in
the otherwise same system of Fig.\,\ref{fig:3} but with lower symbol
rate $R_{s}=\unit[10]{GBd}$, longer link length $L=\unit[4000]{km}$,
and $N_{z}=160$ guard symbols, as shown in Fig.\,\ref{fig:7}b. 

\begin{figure}
\begin{centering}
\includegraphics[width=0.5\columnwidth]{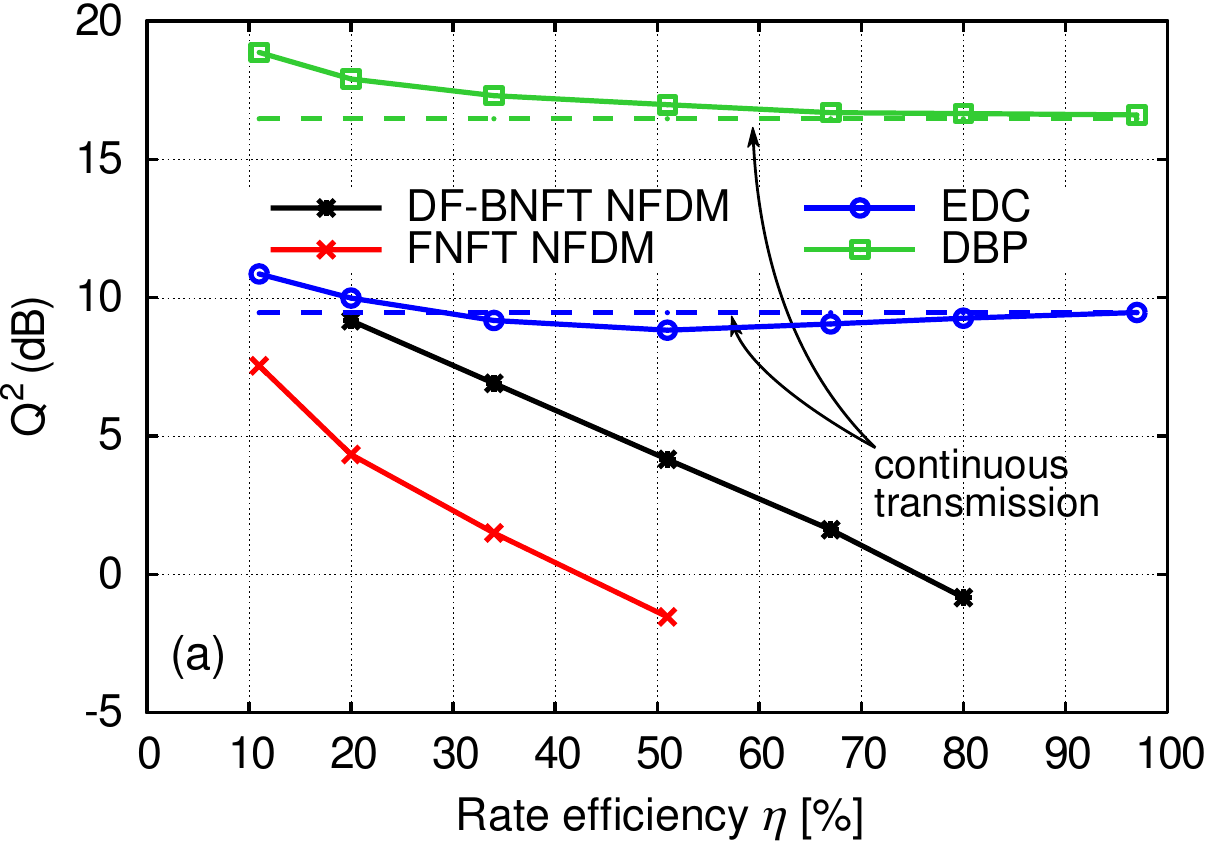}\includegraphics[width=0.5\columnwidth]{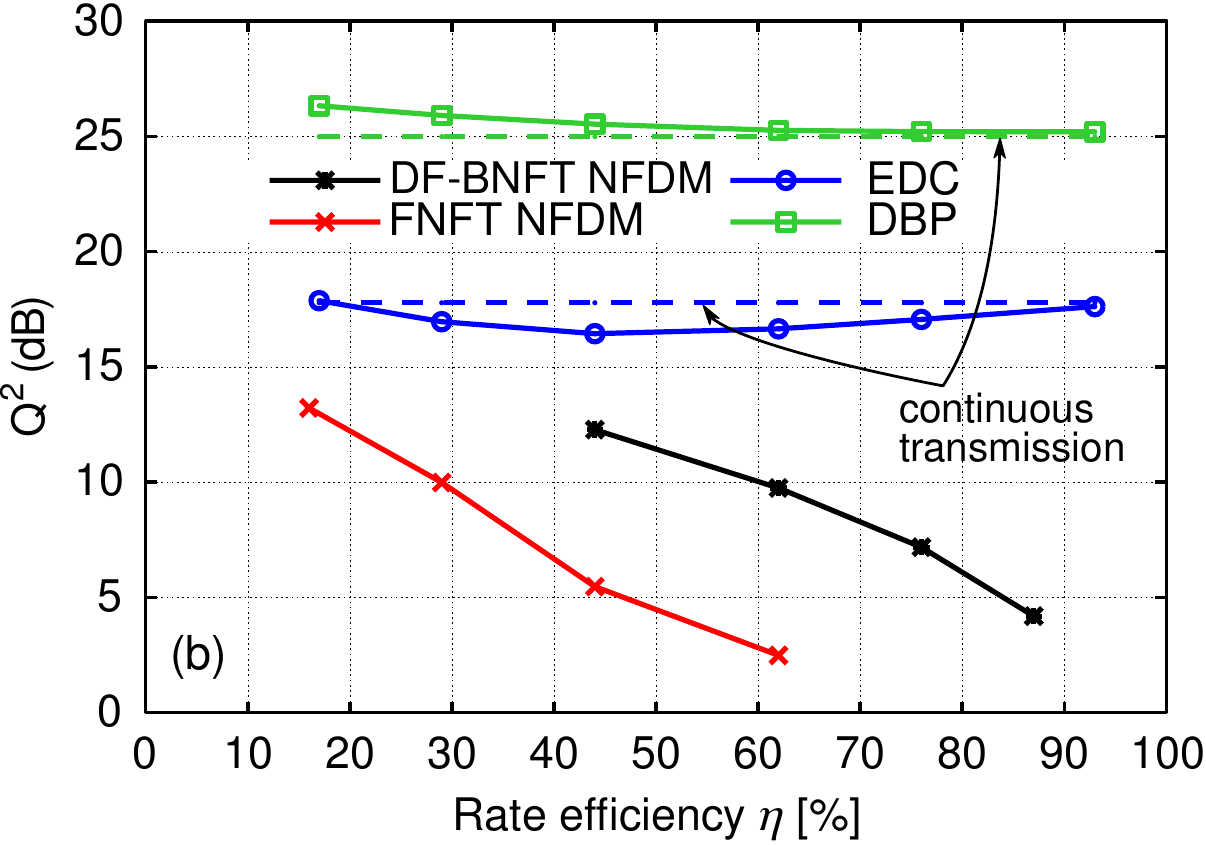}\vspace*{-1ex}
\par\end{centering}
\caption{\label{fig:7}Best achievable performance vs rate efficiency for \protect\ac{NFDM}
with different detection strategies and for conventional systems with
EDC or DBP: (a) low-dispersion fiber with $16$QAM symbols, $\beta_{2}=\unit[-1.27]{ps^{2}/km}$,
$N_{z}=125$, $L=\unit[2000]{km}$, and $R_{s}=\unit[50]{GBd}$; (b)
\protect\ac{QPSK} symbols with $\beta_{2}=\unit[-20.39]{ps^{2}/km}$,
$N_{z}=160$, $L=\unit[4000]{km}$, and $R_{s}=\unit[10]{GBd}$.}
\end{figure}

\begin{figure}
\begin{centering}
\includegraphics[width=0.5\columnwidth]{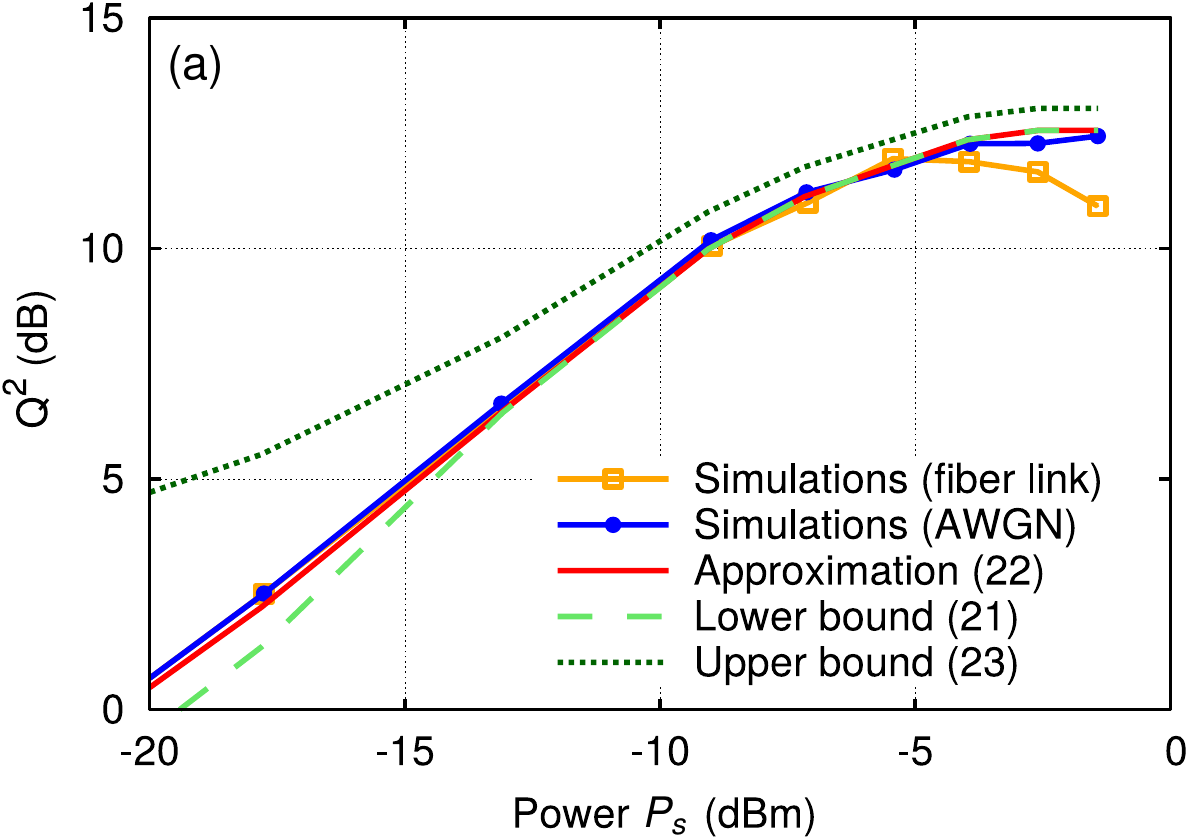}\includegraphics[width=0.5\columnwidth]{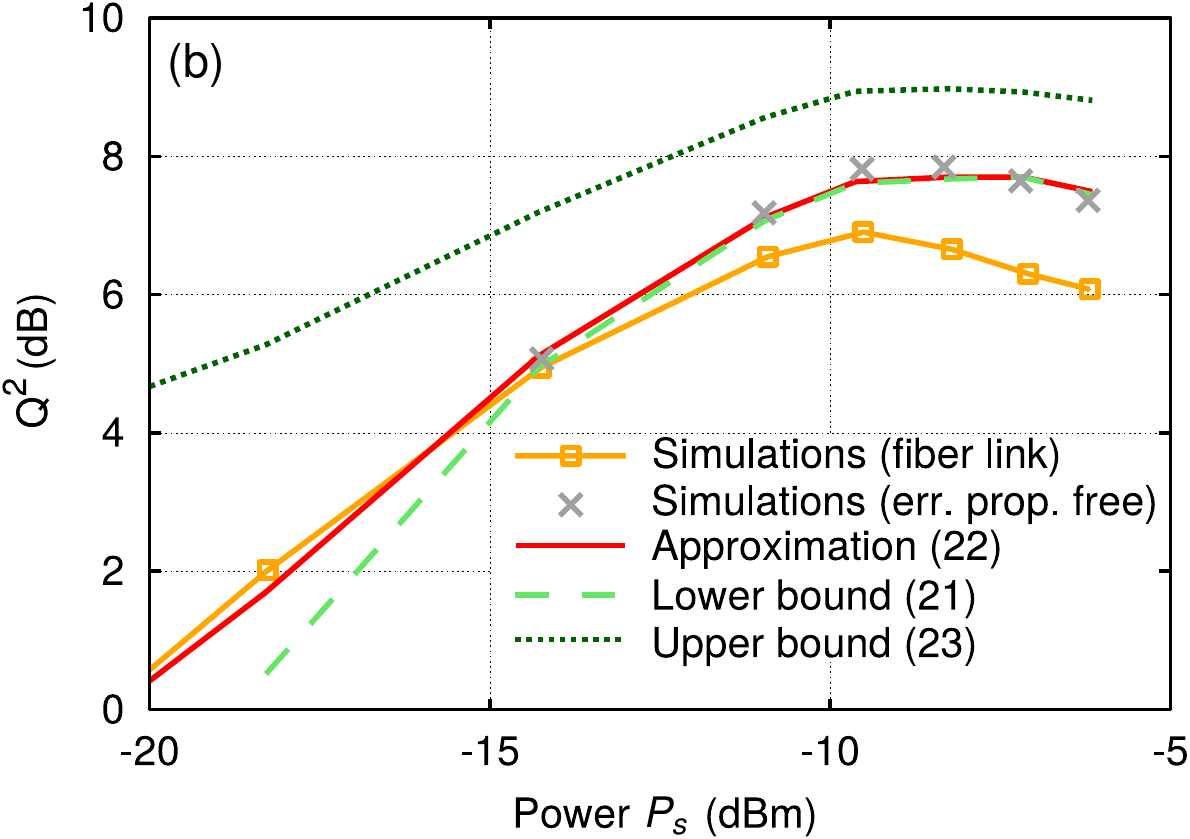}
\vspace*{-1ex}
\par\end{centering}
\caption{\label{fig:8}Validation of the semianalytic approximation and bounds
for the performance of DF-BNFT detection. Same scenario of Fig.~\ref{fig:3},
with (a) $N_{b}=256$ and (b) $N_{b}=1024$.}
\end{figure}

\section{Validation of the approximation and bounds\label{sec:Validation-of-the-bounds}}

The bounds and approximation obtained by replacing (\ref{eq:uppbound})\textendash (\ref{eq:lowbound})
into (\ref{eq:Pe}) are reported in Fig.\,\ref{fig:8} for $N_{b}=256$
and $N_{b}=1024$, after conversion to Q-factor. As the Q-factor is
directly related to the bit-error probability $P_{b}$, the approximation
$P_{b}\simeq P_{e}/M$ is used, assuming that a symbol error always
corresponds to a single bit error. Moreover, if $P_{e}$ increases,
the $Q$ factor decreases, and the other way around. Therefore a lower
(or upper) bound for $P_{e}$ becomes an upper (or lower) bound for
$Q_{\mathrm{dB}}^{2}$. In order to check their accuracy, they are
compared with the performance obtained by numerical simulations for
the actual fiber channel. In both cases, the approximation lies between
the bounds, asymptotically approaching the lower bound when power
increases. At low powers, the approximation is in very good agreement
with numerical simulations. On the other hand, near the optimum power,
the approximation overestimates the actual performance, which falls
slightly below the lower bound. This is due to signal-noise interaction
during propagation (for $N_{b}=256)$ and to error propagation in
the decision-feedback strategy (for $N_{b}=1024)$, both neglected
in the derivation of the bounds and approximation (\ref{eq:uppbound})\textendash (\ref{eq:lowbound}).
In fact, when considering the numerical simulations for the \ac{AWGN}
channel in Fig.~\ref{fig:8}a, and the error-propagation-free simulations
in Fig.~\ref{fig:8}b, they correctly fall between the bounds and
are in excellent agreement with the approximation. 

As already explained, the probability of error for a given sequence
$P_{e}$ in (\ref{eq:Pe}) should be averaged over all possible sequences.
However, the number of possible sequences $M^{N_{b}}$ is practically
unmanageable, making it impossible performing an exact average. Anyway,
most sequences contribute in the same way to the average, so that
we don't need to explore all of them, but only account for the most
significant ones. This can be done by performing a Monte Carlo average,
consisting in randomly generating sequences until the corresponding
average performance stabilizes. This is in contrast with the full
numerical estimation used in Section~\ref{sec:Numerical-results},
in which also the effect of noise is numerically estimated by averaging
over many random realizations. To illustrate the difference between
the two approaches and show the speed of convergence of the various
estimates, Fig.\,\ref{fig:9} reports the same bounds and estimates
shown in Fig.~\ref{fig:8}a as a function of the number of iterations
(corresponding to the number of sequences of length $N_{b}$ over
which the performance is averaged), considering two different values
of the mean power. As can be seen, the computation of the semianalytical
bounds and approximation requires only a few iterations in all cases,
while computing the performance by full numerical simulations requires
a number of iterations that depends on the actual error probability
and, hence, on the input power, such that a sufficient number of error
events are observed. As an example, $30\div40$ iterations suffice
for an input power of $\unit[-9]{dBm}$, while more than 200 iterations
are necessary at the optimum input power of $\unit[-4]{dBm}$. 

\begin{figure}
\begin{centering}
\includegraphics[width=0.6\columnwidth]{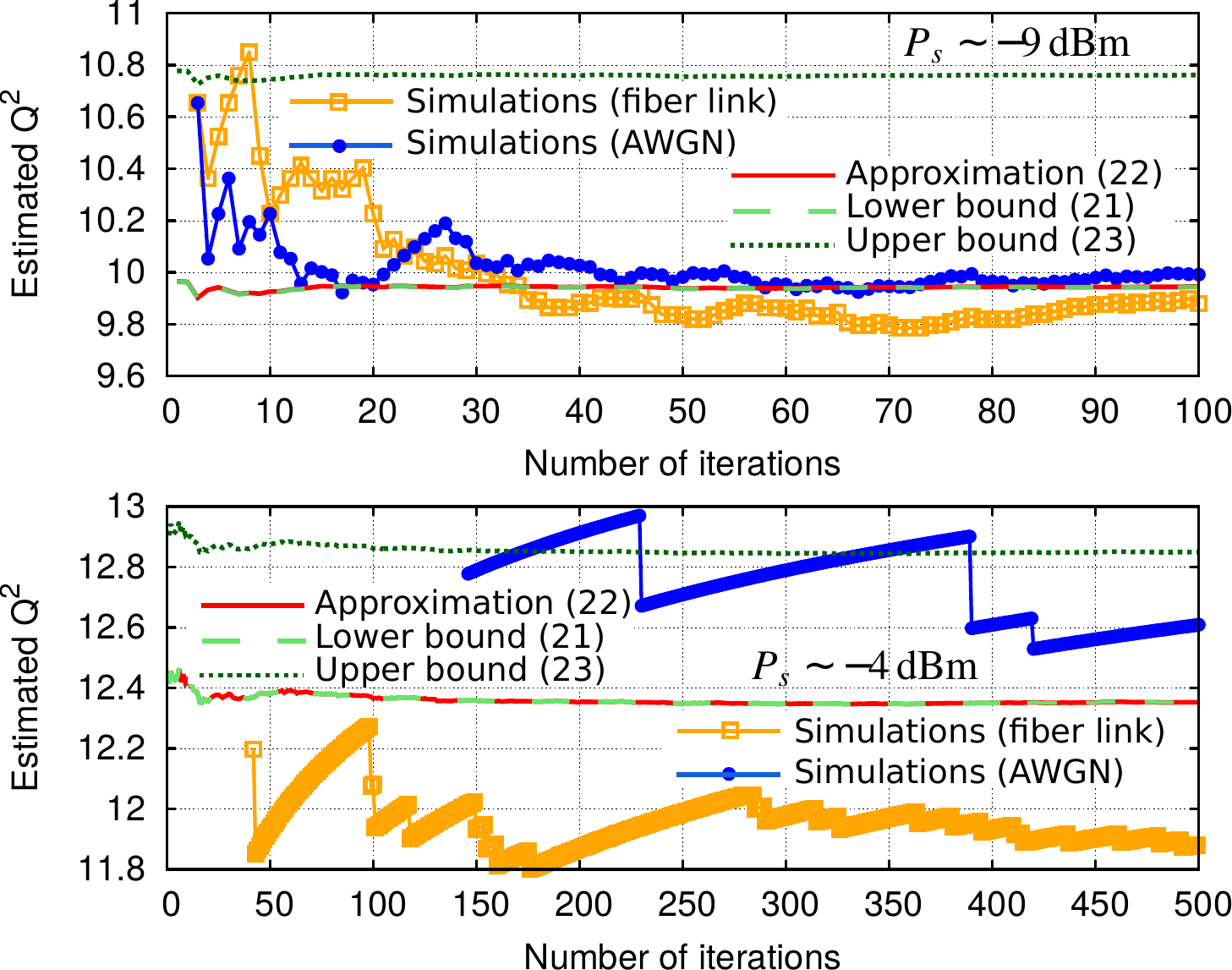}\vspace*{-1ex}
\par\end{centering}
\caption{\label{fig:9}Convergence of the numerical simulations and of the
semianalytic approximation and bounds with the number of iterations
(transmitted sequences). Same scenario of Fig.~\ref{fig:3}, with
$N_{b}=256$ at $P_{s}=-9$dBm (above) and at optimal power $P_{s}=-4$dBm
(below).}
\end{figure}

\section{Conclusions}

Differently than in conventional systems, there is still no theory
for optimum detection when using the \ac{NFT} as a mean to transmit
information. As a result, the existing \ac{NFT}-based transmission
schemes mostly rely on reasonable assumptions and are guided by concepts
and principles borrowed from the more familiar detection theory for
linear systems. This means that big improvements are to be expected
as soon as the principles of nonlinear detection theory are better
understood. There is a lot still to be discovered, but some well known
general principles can be applied independently of the nature of the
systems at hand, be it linear or not. For example, this is the case
for the \ac{MAP} detection strategy to be applied in a given communication
system, once a statistical knowledge of the channel is available or
can be reasonably approximated.

Following this line, by assuming a simple channel model, a novel (suboptimal)
detection strategy for \ac{NFDM} systems, referred to as \ac{DF-BNFT},
has been introduced in this work. In a nutshell, taking advantage
of a peculiar property of the \ac{NFT}, rather than using a strategy
based on the \ac{FNFT} and the minimization of the Euclidean distance
in the nonlinear frequency domain, \ac{DF-BNFT} takes decisions by
minimizing the Euclidean distance in the time domain through \ac{BNFT}
and decision feedback. Moreover, a semianalytical approximation, upper
and lower bounds to system performance have been derived, providing
an effective tool to estimate system performance without resorting
to time-consuming numerical simulations.

As demonstrated by numerical results, the proposed detection strategy
allows for a performance improvement of up to $6.2$\,dB with respect
to standard \ac{FNFT} detection. Despite such a big improvement,
also the performance of the proposed \ac{DF-BNFT} strategy decays
when the burst length increases, similarly to what observed for \ac{FNFT}-based
detection \cite{civelli2017noise}. This peculiar behavior forces
the use of short bursts, separated by long guard times, severely limiting
the overall spectral efficiency achievable by the system. However,
while the performance decays in \ac{FNFT} and \ac{DF-BNFT} detection
appear to be similar, their causes are different: the perturbation
of the nonlinear spectrum caused by optical noise (and non-optimally
accounted for by the detection metrics) in the former case; the information
loss and error propagation caused by the use of the suboptimum strategy
(\ref{eq:suboptimal}) in the latter.

Though the persistence of this undesirable behavior might be daunting,
in fact making \ac{NFDM} not yet competitive with conventional systems,
the relevant performance improvement obtained by a still sub-optimum
strategy is rather encouraging and paves the way for further progresses.
Possible future developments include: the implementation of the optimum
detection strategy (\ref{eq:optimum_detection}) to overcome the limitations
induced by (\ref{eq:suboptimal}); the reduction of computational
complexity; the optimization of the pulse shape to maximize the achievable
spectral efficiency; the investigation of the impact of  more practical
amplification schemes, such as lumped amplification, in which the
lossless path-average model needs to be employed \cite{le2015amplification,le2017naturephotonics};
and the extension of the detection strategy to the dual-polarization
case. Such developments are required to fully exploit the potentials
of the \ac{NFDM} technique and to understand if it can provide any
practical advantages compared to conventional transmission techniques.

\section*{Funding}

Ente Cassa di Risparmio di
Firenze.
\end{document}